\documentclass{llncs}

\usepackage[margin=1in]{geometry}
\usepackage{hyperref}
\usepackage{graphicx}
\usepackage{amsmath}
\usepackage{amsfonts}
\usepackage[ruled,vlined]{algorithm2e}
\usepackage{tikz}
\usepackage{theorem}
\usepackage{color}

\newcommand{\wt}{\::\:}
\newcommand{\st}{\:|\:}
\newcommand{\ltdots}{..}
\DeclareMathOperator*{\argmax}{argmax}

\newcommand{\query}{\mathsf{Query}}
\newcommand{\lasttoreach}{\mathtt{last2reach}}
\newcommand{\forward}{\mathtt{forward}}

\renewcommand{\index}{\mathtt{index}}
\newcommand{\mm}{\mathtt{m}}
\newcommand{\uu}{\mathtt{u}}

\def\EnableMNotes{1}
\ifnum\EnableMNotes=1
\newcommand{\mynote}[1]{\footnote{\sf #1}}
\else
\newcommand{\mynote}[1]{}
\fi

\newtheorem{observation}{Observation}

\newcommand{\alex}[1]{\textcolor{black}{#1}}
\newcommand{\veli}[1]{\textcolor{black}{#1}}

\newcommand{\No}[1]{}

\begin{document}

\title{Using Minimum Path Cover to Boost Dynamic Programming on DAGs: Co-Linear Chaining Extended}

\author{Anna Kuosmanen\inst{1} \and Topi Paavilainen\inst{1} \and Travis Gagie\inst{2} \and Rayan Chikhi\inst{3} \and Alexandru Tomescu\inst{1}$^{,\star}$ \and Veli M{\"a}kinen\inst{1}\thanks{Shared last author contribution}$^{,}$\thanks{Corresponding author: \url{veli.makinen@helsinki.fi}}}

\institute{Helsinki Institute for Information Technology HIIT, 
Department of Computer Science, University of Helsinki, Finland
\and
Diego Portales University, Chile
\and
CNRS, CRIStAL, University of Lille 1, France
}

\maketitle
\thispagestyle{plain}

\begin{abstract}
Aligning sequencing reads on graph representations of genomes is an important ingredient of \emph{pan-genomics} (Marschall et al. \emph{Briefings in Bioinformatics}, 2016). 
Such approaches typically find a set of \emph{local anchors} that indicate plausible matches between substrings of a read to subpaths of the graph. These anchor matches are then combined to form a 
(semi-local) alignment of the complete read on a subpath. \emph{Co-linear chaining} is an algorithmically rigorous approach to combine the anchors. It is a well-known approach for the case of two sequences as inputs. 
Here we extend the approach so that one of the inputs can be a directed acyclic graph (DAGs), e.g. a \emph{splicing graph} in transcriptomics or \emph{variant graph} in pan-genomics. 

The extension of co-linear chaining to DAGs turns out to have a tight connection to the \emph{minimum path cover} problem that asks us to find 
a minimum-cardinality set of paths that cover all the nodes of a DAG. We study the case when the size $k$ of a minimum path cover is small, which is often the case in practice.
First, we propose an algorithm for finding a minimum path cover of a DAG $(V,E)$ in $O(k|E|\log|V|)$ time, improving all known time-bounds when $k$ is small and the DAG is not too dense. An immediate consequence is an improved space/time tradeoff for reachability queries in arbitrary directed graphs.
Second, we introduce a general technique for extending dynamic programming (DP) algorithms from sequences to DAGs. This is enabled by our minimum path cover algorithm, and works by mimicking the DP algorithm for sequences on each path of the minimum path cover. This technique generally produces algorithms that are slower than their counterparts on sequences only by a factor $k$. We illustrate this on two classical problems extended to labeled DAGs: longest increasing subsequence and longest common subsequence. For the former we obtain an algorithm with running time $O(k|E|\log |V|)$. This matches the optimal solution to the classical problem variant when, e.g., the input sequence is modeled as a path. We obtain an analogous result for the longest common subsequence problem.
Finally, we apply this technique to the co-linear chaining problem, that is a generalization of both of the above two problems. The algorithm for this problem turns out to be more involved, needing further ingredients, such as an FM-index tailored for large alphabets, and a two-dimensional range search tree modified to support range maximum queries. We implemented the new co-linear chaining approach. Experiments on splicing graphs show that the new method is efficient also in practice.
\end{abstract}

\section{Introduction\label{sect:intro}}

A \emph{path cover} of a DAG $G = (V,E)$ is a set of paths such that every node of $G$ belongs to some path. A \emph{minimum path cover} (MPC) is one having the minimum number of paths. The size of a MPC is also called the \emph{width} of $G$.
Many DAGs commonly used in genome research, such as graphs encoding human mutations~\cite{variationgraph} and graphs modeling gene transcripts~\cite{splicing_graphs}, can consist, in the former case, of millions of nodes and, in the latter case, of thousands of nodes. However, they generally have a small width on average; for example, \alex{splicing graphs for most genes in human chromosome 2 have width at most 10~\cite[Fig.~7]{TGP2015}.} 
To the best of our knowledge, among the many MPC algorithms~\cite{Fulkerson:1956fk,journals/siamcomp/HopcroftK73,journals/siamcomp/Schnorr78,Orlin:2013:MFO:2488608.2488705,chain-cover,Chen2014}, there are only \alex{three} whose complexities depends on the width of the DAG. Say the width of $G$ is $k$. \alex{The first algorithm runs in time $O(|V||E| + k|V|^2)$ and can be obtained by slightly modifying an algorithm for finding a minimum chain cover in partial orders from \cite{Felsner2003}. The other two algorithms are due to Chen and Chen: the first one} works in time $O(|V|^2 + k\sqrt{k}|V|)$~\cite{chain-cover}, and the second one works in time $O(\max(\sqrt{|V|}|E|,k\sqrt{k}|V|))$~\cite{Chen2014}. 

In this paper we present an MPC algorithm running in time $O(k|E|\log|V|)$. For example, for $k = o(\sqrt{|V|} / \log |V|)$ and $|E| = O(|V|^{3/2})$, this is better than all previous algorithms. Our algorithm is based on the following standard reduction of a minimum flow problem to a maximum flow problem (see e.g.~\cite{Ahuja:1993fk}): (i) find a feasible flow/path cover satisfying all demands, and (ii) solve a maximum flow problem in a graph encoding how much flow can be removed from every edge. Our main insight is to solve step (i) by finding an approximate solution that is greater than the optimal one only by a $O(\log |V|)$ factor. Then, if we solve step (ii) with the Ford-Fulkerson algorithm, the number of iterations can be bounded by $O(k\log|V|)$.

\veli{We then proceed to show that some problems (like pattern matching) that admit efficient \emph{sparse dynamic programming} solutions on sequences \cite{EGRI92} can be extended to DAGs, so that their complexity increases only by the minimum path cover size $k$. Extending pattern matching to DAGs has been studied before \cite{PK95,ALL00,Nav00}. For those edit distance -based formulations our approach does not yield an improvement, but on formulations involving sparse set of matching anchors \cite{EGRI92} we can boost the naive solutions of their DAG extensions by exploiting a path cover. Namely, our improvement applies to many cases where a data structure over previously computed solutions is maintained and queried for computing the next value.}
Our new MPC algorithm enables this, as its complexity is generally of the same form as that of solving the extended problems. 
%!TEX encoding = UTF-8 UnicodeCareful bookkeeping is necessary due to the evaluation order not matching the reachability order among the nodes of the path cover. 
Given a path cover, our technique then computes so-called \emph{forward propagation links} indicating how the partial solutions in each path in the cover must be synchronized.

To best illustrate the versatility of the technique itself, we show (in the Appendix) how to compute a longest increasing subsequence (LIS) in a labeled DAG, in time $O(k |E| \log |V|)$. This matches the optimal solution to the classical problem on a single sequence when, e.g., this is modeled as a path (where $k=1$). We also illustrate our technique with the longest common subsequence (LCS) problem between a labeled DAG $G = (V,E)$ and a sequence $S$.

Finally, we consider the main problem of this paper---co-linear chaining (CLC)---first introduced in~\cite{MW95}. It has been proposed as a model of the sequence alignment problem that scales to massive inputs, and has been a subject of recent interest (see e.g.~\cite{Patro:2017aa,Uricaru2015,Vyverman2015,vyverman2014,DBLP:conf/alcob/WandeltL14,DBLP:journals/bmcbi/MakinenSY12,SK03}). In the CLC problem, the input is directly assumed to be a set of $N$ pairs of intervals in the two sequences that match (either exactly or approximately). The CLC alignment solution asks for a subset of these plausible pairs that maximizes the coverage in one of the sequences, and whose elements appear in increasing order in both sequences. The fastest algorithm for this problem runs in the optimal $O(N \log N)$ time~\cite{Abo07}.

We define a generalization of the CLC problem between a sequence and a labeled DAG. As motivation, we mention the problem of aligning a long sequence, or even an entire chromosome, inside a DAG storing all known mutations of a population with respect to a reference genome (such as the above-mentioned~\cite{variationgraph}\veli{, or more specificly a linearized version of it \cite{Haus17}}). Here, the $N$ input pairs match intervals in the sequence with paths (also called \emph{anchors}) in the DAG. This problem is not straightforward, as the topological order of the DAG might not follow the reachability order between the anchors. 
\No{Explicitly computing the reachabilities between every pair of nodes in $G$ requires $O(|V||E|)$ time, and even computing only the reachabilities between the $N$ anchors requires $O(N|E|)$ time. This already matches the time required for the brute force approach of checking every pair against every other pair.}
Existing tools for aligning DNA sequences to DAGs (BGREAT~\cite{limasset2016read}, vg~\cite{novak2016graph}) rely on anchors but do not explicitly consider solving CLC optimally on the DAG.

The algorithm we propose uses the general framework mentioned above. Since it is more involved, we will develop it in stages. We first give a simple approach to solve a relaxed co-linear chaining problem using $O((|V|+|E|) N)$ time, then introduce the MPC approach that requires $O(k|E| \log |V| + kN \log N)$ time. As above, if the DAG is a labeled path representing a sequence, the running time of our algorithm is reduced to the best current solution for the co-linear chaining problem on sequences, $O(N \log N)$~\cite{Abo07}. We conclude (in the Appendix) with a Burrows-Wheeler technique to efficiently handle a special case that we omitted in this relaxed variant.  We remark that one can reduce the LIS and LCS problems to the CLC problem to obtain the same running time bounds as mentioned earlier; these are given for the sake of comprehensiveness.

In the last section we discuss the anchor-finding preprocessing step. We implemented the new MPC-based co-linear chaining algorithm and conducted experiments on splicing graphs to show that the approach is practical, once anchors are given. Some future directions on how to incorporate practical anchors, and how to apply the techniques to transcript prediction, are discussed. 

\paragraph{Notation.} To simplify notation, for any DAG $G = (V,E)$ we will assume that $V$ is always $\{1,\dots,|V|\}$ and that $1,\dots,|V|$ is a topological order on $V$ (so that for every edge $(u,v)$ we have $u < v$). We will also assume that $|E| \geq |V| - 1$. A \emph{labeled DAG} is a tuple $(V,E,\ell,\Sigma)$ where $(V,E)$ is a DAG and $\ell : V \mapsto \Sigma$ assign to the nodes labels from $\Sigma$, $\Sigma$ being an ordered alphabet.

For a node $v \in V$, we denote by $N^-(v)$ the set of in-neighbors of $v$ and by $N^+(v)$ the set of out-neighbors of $v$. If there is a (possibly empty) path from node $u$ to node $v$ we say that $u$ reaches $v$. We denote by $R^-(v)$ the set of nodes that reach $v$. 
We denote a set of consecutive integers with interval notation $[i..j]$, meaning $\{i,i+1,\ldots,j\}$. For a pair of intervals $m=([x..y],[c..d])$, we use $m.x$, $m.y$, $m.c$, and $m.d$ to denote the four respective endpoints. We also consider pairs of the form $m=(P,[c..d])$ where $P$ is a path, and use $m.P$ to access $P$. The first node of $P$ will be called its \emph{startpoint}, and its last node will be called its \emph{endpoint}. For a set $M$ we may fix an order, to access an element as $M[i]$.

\section{The MPC algorithm}

In this section we assume basic familiarity with network flow concepts; see~\cite{Ahuja:1993fk} for further details. In the \emph{minimum flow problem}, we are given a directed graph $G = (V,E)$ with a single source and a single sink, with a \emph{demand} $d : E \rightarrow \mathbb{Z}$ for every edge. The task is to find a flow of minimum value (the \emph{value} is the sum of the flow on the edges exiting the source) that satisfies all demands (to be called \emph{feasible}). The standard reduction from the minimum path cover problem to a minimum flow one (see, e.g.~\cite{MR545530}) creates a new DAG $G^\ast$ by replacing each node $v$ with two nodes $v^-,v^+$, adds the edge $(v^-,v^+)$ and adds all in-neighbors of $v$ as in-neighbors of $v^-$, and all out-neighbors of $v$ as out-neighbors of $v^+$. Finally, the reduction adds a global source with an out-going edge to every node, and a global sink with an in-coming edge from every node. Edges of type $(v^-,v^+)$ get demand $1$, and all other edges get demand $0$. The value of the minimum flow equals $k$, the width of $G$, and any decomposition of it into source-to-sink paths induces a minimum path cover in $G$.

Our MPC algorithm is based on the following simple reduction of a minimum flow problem to a maximum flow one (see e.g.~\cite{Ahuja:1993fk}): (i) find a feasible flow $f : E \rightarrow \mathbb{Z}$; (ii) transform this into a minimum feasible flow, by finding a maximum flow $f'$ in $G$ in which every $e \in E$ now has capacity $f(e) - d(e)$. The final minimum flow solution is obtained as $f(e) - f'(e)$, for every $e \in E$. Observe that this path cover induces a flow of value $O(k\log|V|)$. Thus, in step (ii) we need to shrink this flow into a flow of value $k$. If we run the Ford-Fulkerson algorithm, this means that there are $O(k\log|V|)$ successive augmenting paths, each of which can be found in time $O(E)$. This gives a time bound for step (ii) of $O(k|E|\log|V|)$.

We solve step (i) in time $O(k|E|\log|V|)$ by finding a path cover in $G^\ast$ whose size is larger than $k$ only by a relative factor $O(\log |V|)$. This is based on the classical greedy set cover algorithm, see e.g.~\cite[Chapter 2]{Vaz01}: at each step, select a path covering most of the remaining uncovered nodes. 

\alex{This approach is similar to the one from \cite{Felsner2003} for finding the minimum number $k$ of chains to cover a partial order of size $n$. A \emph{chain} is a set of pairwise comparable elements. The algorithm from \cite{Felsner2003} runs in time $O(kn^2)$, and it has the same feature as ours: it first finds a set of $O(k\log n)$ chains in the same way as us (longest chains covering most uncovered elements), and then in a second step reduces these to $k$. However, if we were to apply this algorithm to DAGs, it would run in time $O(|V||E| + k|V|^2)$, which is slower than our algorithm for small $k$. This is because it uses the classical reduction given by Fulkerson~\cite{Fulkerson:1956fk} to a bipartite graph, where each edge of the graph encodes a pair of elements in the relation. Since DAGs are not transitive in general, to use this reduction one needs first to compute the transitive closure of the DAG, in time $O(|V||E|)$. Such approximation-refinement approach has also been applied to other covering problems on graphs, such as a 2-hop cover \cite{DBLP:journals/siamcomp/CohenHKZ03}.}

We now show how to solve step (i) within the claimed running time, by dynamic programming.

\begin{lemma}
\label{lemma:approx-path-cover}
Let $G = (V,E)$ be a DAG, and let $k$ be the width of $G$. In time $O(k|E|\log |V|)$, we can compute a path cover $P_1,\dots,P_K$ of $G$, such that $K = O(k\log |V|)$.
\end{lemma}

\begin{proof}
The algorithm works by choosing, at each step, a path that covers the most uncovered nodes. For every node $v \in V$, we store $\mm[v] = 1$, if $v$ is not covered by any path, and $\mm[v] = 0$ otherwise. We also store $\uu[v]$ as the largest number of uncovered nodes on a path starting at $v$. The values $\uu[\cdot]$ are computed by dynamic programming, by traversing the nodes in inverse topological order and setting $\uu[v] = \mm[v] + \max_{w \in N^+(v)} \uu[v]$. 
Initially we have $\mm[v] = 1$ for all $v$. We then compute $\uu[v]$ for all $v$, in time $O(|E|)$. By taking the node $v$ with the maximum $\uu[v]$, and tracing back along the optimal path starting at $v$, we obtain our first path in time $O(|E|)$. We then update $\mm[v] = 0$ for all nodes on this path, and iterate this process until all nodes are covered. This takes overall time $O(K|E|)$, where $K$ is the number of paths found.

This algorithm analysis is identical to the one of the classical greedy set cover algorithm~\cite[Chapter 2]{Vaz01}, because the universe to be covered is $V$ and each possible path in $G$ is a possible covering set, which implies that $K = O(k\log |V|)$.
\qed
\end{proof}

Combining Lemma~\ref{lemma:approx-path-cover} with the above-mentioned application of the Ford-Fulkerson algorithm, we obtain our first result:

\begin{theorem}
\label{thm:MPC}
Given a DAG $G = (V,E)$ of width $k$, the MPC problem on $G$ can be solved in time $O(k|E|\log|V|)$.
\end{theorem}

\section{The dynamic programming framework}
\label{sec:framework}

In this section we give an overview of the main ideas of our approach. 

Suppose we have a problem involving DAGs that is solvable, for example by dynamic programming, by traversing the nodes in topological order. Thus, assume also that a partial solution at each node $v$ is obtainable from all (and only) nodes of the DAG that can reach $v$, plus some other independent objects, such as another sequence. Furthermore, suppose that at each node $v$ we need to query (and maintain) a data structure $\mathcal{T}$ that depends on $R^-(v)$ and such that the answer $\query(R^-(v))$ at $v$ is decomposable as:
\begin{equation}
\query(R^-(v)) = \bigoplus_i \query(R^-_i(v)).
\label{eq:decomposition}
\end{equation}
In the above, the sets $R^-_i(v)$ are such that $R^-(v) = \bigcup_i R^-_i(v)$, they are not necessarily disjoint, and $\bigoplus$ is some operation on the queries, such as min or max, that does not assume disjointness. It is understood that after the computation at $v$, we need to update $\mathcal{T}$. It is also understood that once we have updated $\mathcal{T}$ at $v$, we cannot query $\mathcal{T}$ for a node before $v$ in topological order, because it would give an incorrect answer.

The first idea is to decompose the graph into a path cover $P_1,\dots,P_K$. As such, we decompose the computation only along these paths, in light of (\ref{eq:decomposition}). We replace a single data structure $\mathcal{T}$ with $K$ data structures $\mathcal{T}_1,\dots,\mathcal{T}_K$, and perform the operation from (\ref{eq:decomposition}) on the results of the queries to these $K$ data structures.

Our second idea concerns the order in which the nodes on these $K$ paths are processed. Because the answer at $v$ depends on $R^-(v)$, we cannot process the nodes on the $K$ paths (and update the corresponding $\mathcal{T}_i$'s) in an arbitrary order. As such, for every path $i$ and every node $v$, we distinguish the \emph{last node} on path $i$ that reaches $v$ (if it exists). We will call this node $\lasttoreach[v,i]$. See Figure~\ref{fig:propagation-links} for an example. We note that this insight is the same as in \cite{Jagadish:1990:CTM:99935.99944}, which symmetrically identified the \emph{first} node on a chain $i$ that can be reached from $v$ (a \emph{chain} is a subsequence of a path). The following observation is the first ingredient for using the decomposition (\ref{eq:decomposition}).

\begin{figure}[t]
\centering
\includegraphics[width=8cm]{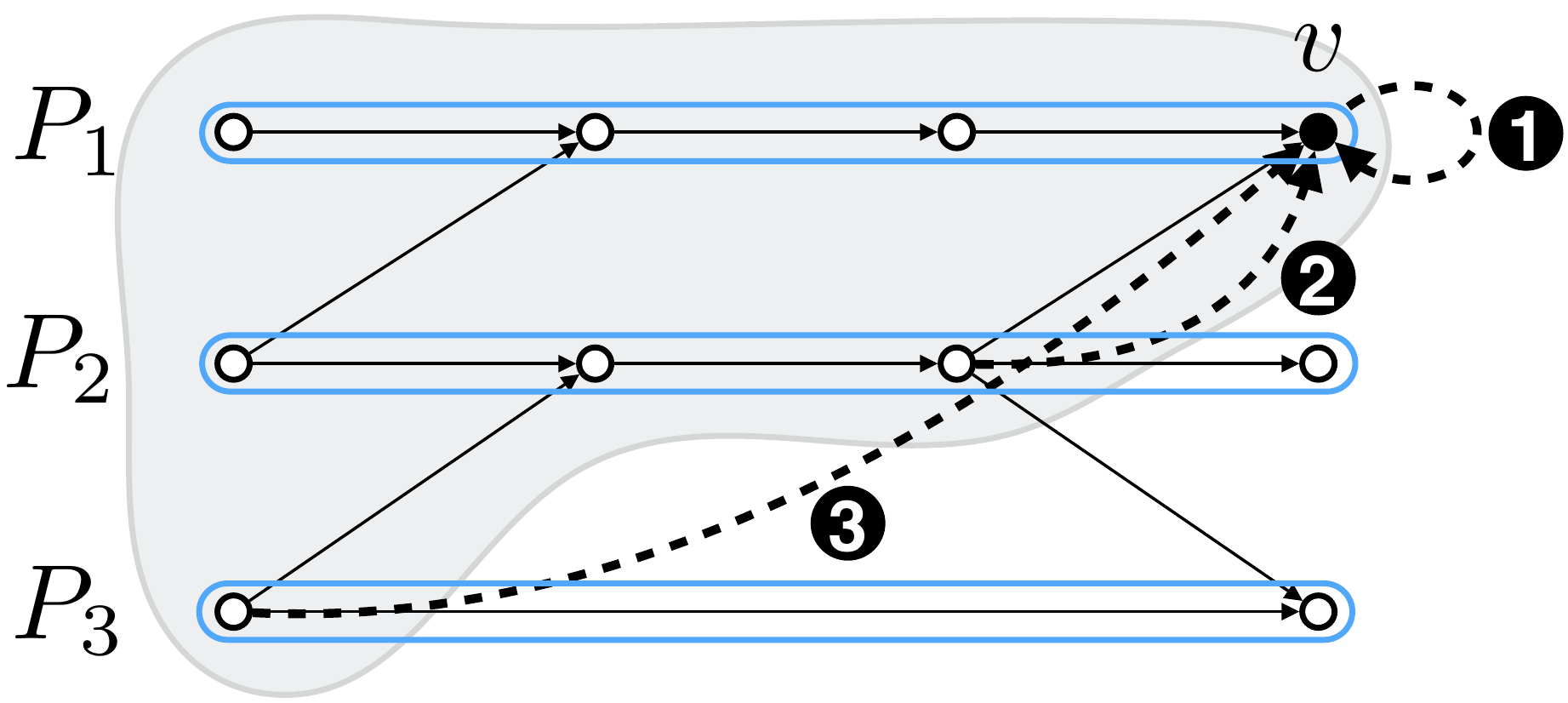}
    \caption{A path cover $P_1,P_2,P_3$ of a DAG. The forward links entering $v$ from $last2reach[v,i]$ are shown with dotted black lines, for $i\in\{1,2,3\}$. We mark in gray the set $R^-(v)$ of nodes that reach $v$.\label{fig:propagation-links}}
\end{figure}

\begin{observation}
Let $P_1,\dots,P_K$ be a path cover of a DAG $G$, and let $v \in V(G)$. Let $R_i$ denote the set of nodes of $P_i$ from its beginning until $\lasttoreach[v,i]$ inclusively (or the empty set, if $\lasttoreach[v,i]$ does not exist). Then $R^-(v) = \bigcup_{i=1}^{K}R_i$.
\end{observation}
\begin{proof}
It is clear that $\bigcup_{i=1}^{K} R_i \subseteq R^-(v)$. To show the reverse inclusion, consider a node $u \in R^-(v)$. Since $P_1,\dots,P_K$ is a path cover, then $u$ appears on some $P_i$. Since $u$ reaches $v$, then $u$ appears on $P_i$ before $\lasttoreach[v,i]$, or $u = \lasttoreach[v,i]$. Therefore $u$ appears on $R_i$, as desired.
\qed
\end{proof}

This allows us to identify, for every node $u$, a set of \emph{forward propagation links} $\forward[u]$, where $(v,i) \in \forward[u]$ holds for any node $v$ and index $i$ with $\lasttoreach[v,i] = u$. These propagation links are the second ingredient in the correctness of the decomposition. Once we have computed the correct value at $u$, we update the corresponding data structures $\mathcal{T}_i$ for all paths $i$ to which $u$ belongs. We also propagate the query value of $\mathcal{T}_i$ in the decomposition (\ref{eq:decomposition}) for all nodes $v$ with $(v,i) \in \forward[u]$. This means that when we come to process $v$, we have already correctly computed all terms in the decomposition (\ref{eq:decomposition}) and it suffices to apply the operation $\bigoplus$ to these terms.

The next lemma shows how to compute the values $\lasttoreach$ (and, as a consequence, all forward propagation links), also by dynamic programming.

\begin{lemma}
\label{lemma:last2reach}
Let $G = (V,E)$ be a DAG, and let $P_1,\dots,P_K$ be a path cover of $G$. For every $v \in V$ and every $i \in [1..K]$, we can compute $\lasttoreach[v,i]$ in overall time $O(K|E|)$.
\end{lemma}
\begin{proof}
For each $P_i$ and every node $v$ on $P_i$, let $\index[v,i]$ denote the position of $v$ in $P_i$ (say, starting from $1$). Our algorithm actually computes $\lasttoreach[v,i]$ as the index of this node in $P_i$. Initially, we set $\lasttoreach[v,i] = -1$ for all $v$ and $i$. At the end of the algorithm, $\lasttoreach[v,i] = -1$ will hold precisely for those nodes $v$ that cannot be reached by any node of $P_i$. We traverse the nodes in topological order. For every $i \in [1..K]$, we do as follows: if $v$ is on $P_i$, then we set $\lasttoreach[v,i] = \index[v,i]$. Otherwise, we compute by dynamic programming $\lasttoreach[v,i]$  as $\max_{u \in N^-(v)} \lasttoreach[u,i]$.
\qed
\end{proof}

An immediate application of Theorem~\ref{thm:MPC} and of the values $\lasttoreach[v,i]$ is for solving reachability queries (see Appendix~\ref{appendix:reachability}). Another simple application is an extension of the \emph{longest increasing subsequence (LIS)} problem to labeled DAGs (Appendix~\ref{appendix:LIS}). 

The LIS problem, the LCS problem of Section~\ref{sect:LCS}, as well as the CLC problem of Section~\ref{sect:colinearchaining} make use of the following standard data structure (see e.g.~\cite[p.20]{MBCT15}).

\begin{lemma}
\label{lemma:searchtree}
The following two operations can be supported with a balanced binary search tree $\mathcal{T}$ in time $O(\log n)$, where $n$ is the number of leaves in the tree.
\begin{itemize}
\item $\mathsf{update}(k,\mathtt{val})$: For the leaf $w$ with $\mathtt{key}(w)=k$, update $\mathtt{value}(w)=\mathtt{val}$.
\item $\mathsf{RMaxQ}(l,r)$: Return $\max_{w \wt l\leq \mathtt{key}(w) \leq r} \mathtt{value}(w)$ (\emph{Range Maximum Query}).
\end{itemize}
Moreover, the balanced binary search tree can be built in $O(n)$ time, given the $n$ pairs $(\mathtt{key},\mathtt{value})$ sorted by component $\mathtt{key}$. 
\end{lemma}

\section{The LCS problem \label{sect:LCS}}

Consider a labeled DAG $G=(V,E,\ell,\Sigma)$ and a sequence $S \in \Sigma^*$, where $\Sigma$ is an ordered alphabet. We say that the \emph{longest common subsequence (LCS)} between $G$ and $S$ is a longest subsequence $C$ of any path label in $G$ such that $C$ is also a subsequence of $S$. 

We will modify the LIS algorithm of Appendix~\ref{appendix:LIS} minimally to find a LCS between a DAG $G$ and a sequence $S$. 
The description is self-contained yet, for the interest of page limit, more dense than the LIS algorithm derivation. The purpose is to give an example of the general MPC-framework with fewer technical details than required in the main result of this paper concerning co-linear chaining.

%Let $\mathtt{select}_c(S,j')$ give the index $j$ with the $j'$-th occurrence of symbol $c$ in $S$, for $1\leq j'\leq \mathtt{occ}_c(S)$, where $\mathtt{occ}_c(S)$ denotes the number of occurrences of $c$ in $S$.
For any $c \in \Sigma$, let $S(c)$ denote set $\{j \mid S[j]=c\}$. For each node $v$ and each $j\in S(\ell(v))$, we aim to store in $\mathtt{LLCS}[v,j]$ the length of the longest common subsequence between $S[1..j]$ and any label of path ending at $v$, among all subsequences having $\ell(v)=S[j]$ as the last symbol. 

Assume we have a path cover of size $K$ and $\forward[u]$ computed for all $u\in V$. Assume also we have mapped $\Sigma$ to $\{0,1,2,\ldots,|S|+1\}$ in $O((|V|+|S|) \log |S|)$ time (e.g. by sorting the symbols of $S$, binary searching labels of $V$, and then relabeling by ranks, with the exception that, if a node label does not appear in $S$, it is replaced by $|S|+1$). 

Let $\mathcal{T}_i$ be a search tree of Lemma~\ref{lemma:searchtree} initialized with key-value pairs $(0,0)$, $(1,-\infty)$, $(2,-\infty)$, \ldots, $(|S|,-\infty)$, for each $i \in [1..K]$. The algorithm proceeds in fixed topological ordering on $G$. At a node $u$, for every $(v,i) \in \forward[u]$ we now update an array $\mathtt{LLCS}[v,j]$ for all $j \in S(\ell(v))$ as follows: $\mathtt{LLCS}[v,j]=\max(\mathtt{LLCS}[v,j],\mathcal{T}_i.\mathsf{RMaxQ}(0,j-1)+1)$. The update step of $\mathcal{T}_i$ when the algorithm reaches a node $v$, for each covering path $i$ containing $v$, is done as $\mathcal{T}_{i}.\mathsf{update}(j',\mathtt{LLCS}[v,j'])$ for all $j'$ with $j'<j$ and $j' \in S(\ell(v))$. Initialization is handled by the $(0,0)$ key-value pair so that any $(v,j)$ with $\ell(v)=S[j]$ can start a new common subsequence. 

The final answer to the problem is $\max_{v\in V, j\in S(\ell(v))} \mathtt{LLCS}[v,j]$, with the actual LCS to be found with a standard traceback. The algorithm runs in $O((|V|+|S|)\log |S|+K|M| \log |S|)$ time, where $M=\{(v,j) \mid v \in V, j \in [1..|S|], \ell(v)=S[j]\}$, and assuming a cover of $K$ paths is given. \alex{Notice that $|M|$ can be $\Omega(|V||S|)$.} With Theorem~\ref{thm:MPC} plugged in, the total running time becomes $O(k|E| \log |V| + (|V|+|S|)\log |S|+k|M| \log |S|)$. Since the queries on the data structures are semi-open, one can use the more efficient data structure from ~\cite{Gabow:1984:SRT:800057.808675} to improve the bound to $O(k|E| \log |V| + (|V|+|S|)\log |S|+k|M| \log \log |S|)$. The following theorem summarizes this result.

\begin{theorem}
Let $G = (V,E,\ell,\Sigma)$ be a labeled DAG of width $k$, and let $S \in \Sigma^\ast$, where $\Sigma$ is an ordered alphabet. We can find a longest common subsequence between $G$ and $S$ in time $O(k|E| \log |V| + (|V|+|S|)\log |S|+k|M| \log \log |S|)$.
\end{theorem}

When $G$ is a path, the bound improves to $O((|V|+|S|)\log |S|+|M|\log\log |S|)$, which nearly matches the fastest sparse dynamic programming algorithm for the LCS on two sequences \cite{EGRI92} (with a difference in $\log \log $-factor due to a different data structure, which does not work for this order of computation). 

\section{Co-linear chaining\label{sect:colinearchaining}}

We start with a formal definition of the co-linear chaining problem (see Figure~\ref{fig:colinearchainingoutput} for an illustration), following the notions introduced in~\cite[Section 15.4]{MBCT15}.

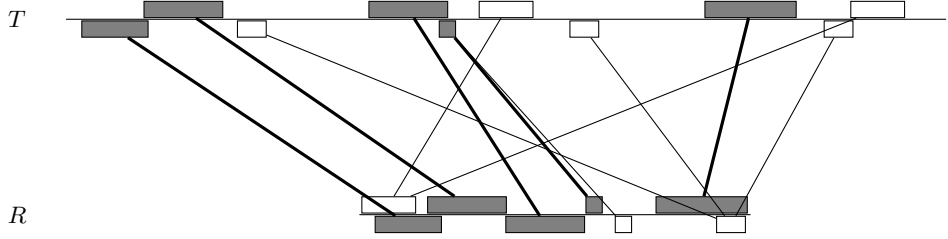
\begin{figure}[t]
\begin{center}
\begin{tikzpicture}[scale=1.3]
\node[circle,draw=none,fill=none] (T) at (-0.5,1) {$T$};
\node[circle,draw=none,fill=none] (R) at (-0.5,-1) {$R$};
\draw (0,1) -- (9,1);
\draw (3,-1) -- (7,-1);
\node[rectangle,draw=black,fill=gray] (A) at (0.5,0.9) {\verb+     +};
\node[rectangle,draw=black,fill=gray] (C) at (1.2,1.1) {\verb+      +};
\node[rectangle,draw=black,fill=none] (D) at (1.9,0.9) {\verb+  +};
\node[rectangle,draw=black,fill=gray] (E) at (3.5,1.1) {\verb+      +};
\node[rectangle,draw=black,fill=gray] (F) at (3.9,0.9) {\verb+ +};
\node[rectangle,draw=black,fill=none] (G) at (4.5,1.1) {\verb+    +};
\node[rectangle,draw=black,fill=none] (H) at (5.3,0.9) {\verb+  +};
\node[rectangle,draw=black,fill=gray] (I) at (7.0,1.1) {\verb+       +};
\node[rectangle,draw=black,fill=none] (J) at (7.9,0.9) {\verb+  +};
\node[rectangle,draw=black,fill=none] (K) at (8.3,1.1) {\verb+    +};

\node[rectangle,draw=black,fill=none] (RKG) at (3.3,-0.9) {\verb+    +};
\node[rectangle,draw=black,fill=gray] (RA) at (3.5,-1.1) {\verb+     +};
\node[rectangle,draw=black,fill=gray] (RC) at (4.1,-0.9) {\verb+      +};
\node[rectangle,draw=black,fill=gray] (RE) at (4.9,-1.1) {\verb+      +};
\node[rectangle,draw=black,fill=gray] (RF) at (5.4,-0.9) {\verb+ +};
\node[rectangle,draw=black,fill=none] (RL) at (5.7,-1.1) {\verb+ +};
\node[rectangle,draw=black,fill=gray] (RI) at (6.5,-0.9) {\verb+       +};
\node[rectangle,draw=black,fill=none] (RM) at (6.8,-1.1) {\verb+  +};

\draw[very thick] (A) to (RA);
\draw[very thick] (C) to (RC);
\draw[very thick] (E) to (RE);
\draw[very thick] (F) to (RF);
\draw[very thick] (I) to (RI);
\draw (D) to (RM);
\draw (F) to (RL);
\draw (H) to (RM);
\draw (J) to (RM);
\draw (K) to (RKG);
\draw (G) to (RKG);

\end{tikzpicture}
\caption{\alex{In the co-linear chaining problem between two sequences $T$ and $R$, we need to find a subset of pairs of intervals (i.e., anchors) so that (i) the selected intervals in each sequence appear in increasing order; and (ii) the selected intervals cover in $R$ the maximum amount of positions. The figure shows an input for the problem, and highlights in gray an optimal subset of anchors. Figure taken from~\cite{MBCT15}.\label{fig:colinearchainingoutput}}}
\end{center}
\end{figure}  

\begin{problem}[Co-linear chaining (CLC)\label{prob:colinearchaining}]
Let $T$ and $R$ be two sequences over an alphabet $\Sigma$, and let $M$ be a set of $N$ pairs $([x \ltdots y],[c \ltdots d])$. Find an ordered subset $S=s_1 s_2 \cdots s_p$ of pairs from $M$ such that
\begin{itemize}
\item $s_{j-1}.y<s_{j}.y$ and $s_{j-1}.d<s_{j}.d$, for all $1\leq j \leq p$, and
\item $S$ maximizes the \emph{ordered} coverage of $R$, defined as \[\mathtt{coverage}(R,S)=|\{i \in [1\ltdots|R|] \st i \in [s_{j}.c \ltdots s_{j}.d] \text{ for some } 1\leq j \leq p\}|.\]
\end{itemize}
\ \vspace{-.4cm}
\end{problem}

%The CLC problem can be solved efficiently using dynamic programming~\cite{SK03}. We will give an overview of this algorithm (following its description from \cite{Abo07}), because it helps in understanding our main result. Before proceeding to it, let us introduce its extension to DAGs.

The definition of ordered coverage between two sequences is symmetric, as we can simply exchange the roles of $T$ and $R$. But when solving the CLC problem between a DAG and a sequence, we must choose whether we want to maximize the ordered coverage on the sequence $R$ or on the DAG $G$. We will consider the former variant.

First, we define the following \emph{precedence relation}:

\begin{definition}
Given two paths $P_1$ and $P_2$ in a DAG $G$, we say that $P_1$ \emph{precedes} $P_2$, and write $P_1 \prec P_2$, if one of the following conditions holds:
\begin{itemize}
\item $P_1$ and $P_2$ do not share nodes and there is a path in $G$ from the endpoint of $P_1$ to the startpoint of $P_2$, or
\item $P_1$ and $P_2$ have a suffix-prefix overlap and $P_2$ is not fully contained in $P_1$; that is, if $P_1 = (a_1,\dots,a_i)$ and $P_2 = (b_1,\dots,b_j)$ then there exists a $k \in \{\max(1,2+i-j),\dots,i\}$ such that $a_k = b_1$, $a_{k+1} = b_2$, \dots, $a_{i} = b_{1+i-k}$.
\end{itemize}
\end{definition}
We then extend the formulation of Problem~\ref{prob:colinearchaining} to handle a sequence and a DAG.

%\begin{problem}[Co-linear chaining between a sequence and a DAG\label{prob:colinearchainingDAG}]
\begin{problem}[CLC between a sequence and a DAG\label{prob:colinearchainingDAG}]
Let $R$ be a sequence, let $G$ be a labeled DAG, and let $M$ be a set of $N$ pairs $(P,[c \ltdots d])$, where $P$ is a path in $G$ and $c \leq d$ are non-negative integers. Find an ordered subset $S=s_1 s_2 \cdots s_p$ of pairs from $M$ such that
\begin{itemize}
\item for all $2 \leq j \leq p$, it holds that $s_{j-1}.P \prec s_{j}.P$ and $s_{j-1}.d < s_{j}.d$, and
\item \begin{sloppypar}
$S$ maximizes the \emph{ordered} coverage of $R$, analogously defined as $\mathtt{coverage}(R,S)=|\{i \in [1\ltdots|R|] \st i \in [s_{j}.c \ltdots s_{j}.d] \text{ for some } 1\leq j \leq p\}|$.
\end{sloppypar}
\end{itemize}
\ \vspace{-.4cm}
\end{problem}

To illustrate the main technique of this paper, let us for now only seek solutions where paths in consecutive pairs in a solution do not overlap in the DAG. Suffix-prefix overlaps between paths turn out to be challenging; we will postpone this case until Appendix~\ref{sect:pathoverlaps}. 

%\begin{problem}[Overlap-limited co-linear chaining between a sequence and a DAG\label{prob:colinearchainingDAGlimited}]
\begin{problem}[Overlap-limited CLC between a sequence and a DAG\label{prob:colinearchainingDAGlimited}]
Let $R$ be a sequence, let $G$ be a labeled DAG, and let $M$ be a set of $N$ pairs $(P,[c \ltdots d])$, where $P$ is a path in $G$ and $c \leq d$ are non-negative integers (with the interpretation that $\ell(P)$ matches $R[c \ltdots d]$). Find an ordered subset $S=s_1 s_2 \cdots s_p$ of pairs from $M$ such that
\begin{itemize}
\item for all $2 \leq j \leq p$, it holds that there is a non-empty path from the last node of $s_{j-1}.P$ to the first node of $s_{j}.P$ and $s_{j-1}.d < s_{j}.d$, and
\item $S$ maximizes $\mathtt{coverage}(R,S)$. 
\end{itemize}
\ \vspace{-.4cm}
\end{problem}

%In what follows, we do not explicitly analyse cases where intervals or paths are included in one another; one can show that the algorithms are implicitly also correctly taking these into account, just like in \cite{Abo07} for the case of sequences.  

%For calculating the maximum coverage we can use dynamic programming to compute the maximum coverage using pair $M[j]$ and any combination of pairs preceding it. As with the case between two sequences (to be considered in the sequel), we have to consider two cases: (a) $M[j']$ does not overlap with $M[j]$ in $R$ or (b) $M[j']$ overlaps with $M[j]$ in $R$.

%However, when one of the inputs is a graph, we cannot simply enumerate over all the tuples $v_1,v_2,\dots,v_{j-1}$ when calculating the coverage for the tuple $v_j$. Assume we have topologically ordered a graph (topological sort can be done in linear time with several algorithms). The relation of paths $P_i$ and $P_j$ is not as simple as comparing their start and end coordinates in this topological order, as path $P_i$ might not be a predecessor of path $P_j$ (for an example see Figure~\ref{fig:precedence}).

%\begin{figure}
%\includegraphics{images/precedence.pdf}
%\caption{\label{fig:precedence}
%Even though the start node of path \{6,7,8\} (red) is larger than the end node of path \{2,3,4,5\} (blue), blue path is not a predecessor of red path as there is no path from last blue node to first red node.}
%\end{figure}

First, let us consider a trivial approach to solve Problem~\ref{prob:colinearchainingDAGlimited}. Assume we have ordered in $O(|E| + N)$ time the $N$ input pairs as $M[1],M[2],\dots, M[N]$, so that the endpoints of $M[1].P, M[2].P,\dots,M[N].P$ are in topological order, breaking ties arbitrarily. We denote by $C[j]$ the maximum ordered coverage of $R[1 \ltdots M[j].d]$ using the pair $M[j]$ and any subset of pairs from $\{M[1],M[2],\dots, M[j-1]\}$.

\begin{theorem}
\label{thm:colinearchainingDAGtrivial}
Overlap-limited co-linear chaining between a sequence and a labeled DAG $G=(V,E,\ell,\Sigma)$ (Problem~\ref{prob:colinearchainingDAGlimited}) on $N$ input pairs can be solved in $O((|V| + |E|) N)$ time.
\end{theorem}

\begin{proof}
First, we reverse the edges of $G$. Then we mark the nodes that correspond to the path endpoints for every pair. After this preprocessing we can start computing the maximum ordered coverage for the pairs as follows: for every pair $M[j]$ in topological order of their path endpoints for $j \in \{1,\dots,N\}$ we do a depth-first traversal starting at the startpoint of path $M[j].P$. Note that since the edges are reversed, the depth-first traversal checks only pairs whose paths are predecessors of $M[j].P$. 

Whenever we encounter a node that corresponds to the path endpoint of a pair $M[j']$, we first examine whether it fulfills the criterion $M[j'].d < M[j].c$ (call this case (a)). 
The best ordered coverage using pair $M[j]$ after all such $M[j']$ is then
\begin{equation}
C^\textrm{a}[j]=\max_{j' \wt M[j'].d<M[j].c} \{ C[j']+(M[j].d-M[j].c+1) \}, \label{eq:case-a}
\end{equation} 
where $C[j]'$ is the best ordered coverage when using pairs $M[j']$ last.

If pair $M[j']$ does not fulfill the criterion for case (a), we then check whether $M[j].c \leq M[j'].d \leq M[j].d$ (call this case (b)).
The best ordered coverage using pair $M[j]$ after all such $M[j']$ with $M[j'].c < M[j].c$ is then
\begin{equation}
C^\textrm{b}[j]=\max_{j' \wt M[j].c\leq M[j'].d\leq M[j].d} \{C[j']+(M[j].d-M[j'].d)\}.\label{eq:case-b}
\end{equation}
Inclusions, i.e. $M[j].c \leq M[j'].c$, can be left computed incorrectly in $C^\textrm{b}[j]$, since there is a better or equally good solution computed in $C^\textrm{a}[j]$ or $C^\textrm{b}[j]$ that does not use them~\cite{Abo07}. 

Finally, we take $C[j]=\max(C^\textrm{a}[j],C^\textrm{b}[j])$. Depth-first traversal takes $O(|V|+|E|)$ time and is executed $N$ times, for $O((|V| + |E|) N)$ total time.
\qed
\end{proof}

%\begin{algorithm}
%\dontprintsemicolon
%\KwIn{A graph $G$ and a set of pairs $V$ in the form $(P,[c \ltdots d])$.}
%\KwOut{The index $j$ giving $\max_j C[j]$.}
%\For{$j \gets 1$ to $N$}{
%	mark $G[j.P.last]$ as the endpoint of tuple $j$\;
%}
%Reverse the arcs of $G$\;
%\For{$j \gets 1$ to $N$}{
%  maxcov $\gets 0$\;
%  $u \gets G[j.P.first]$\;
%  \For{$w$ in $u$.depth-first-traversal} {
 %   \uIf{$w$.marked}{
  %    $j' \gets w$.tuple\;
  %    $C \gets -1$\;
  %    \uIf{$v_{j'}.d < v_j.c$}{
%        $C \gets C[j']+(v_{j}.d-v_{j}.c+1)$\;
 %     }
  %    \uElseIf{$v_j.c \leq v_{j'}.c \leq v_j.d$}{
  %      $C \gets C[j']+(v_{j}.d-v_{j'}.d)$\;
  %    }

%      \uIf{$(v_j.d > v_{j'}.d$}{
%        $C^\textrm{a} \gets C[j']+(v_{j}.d-v_{j}.c+1)$\;
%        $C^\textrm{b} \gets C[j']+(v_{j}.d-v_{j'}.d)$\;
%        $C \gets \max(C^\textrm{a}, C^\textrm{b})$\;
  %      \uIf{$C>\textrm{maxcov}$}{
    %       $\textrm{maxcov} \gets C$\;
      %  }
%      }
   % }
 % }
 % $C[j] \gets \textrm{maxcov}$\;
%}
%\Return{$\argmax_j C[j]$}\;
%\caption{\label{algo:colinearchainingDAGtrivial}
%Trivial algorithm for co-linear chaining between a sequence and a graph. Function $v$.tuple returns the tuple that marked this node as its endpoint.}
%\end{algorithm}

However, we can do significantly better than $O((|V| + |E|) N)$ time. In the next sections we will describe how to apply the framework from Section~\ref{sec:framework} here. 

\subsection{Co-linear chaining on sequences revisited}
 
We now describe the dynamic programming algorithm from~\cite{Abo07} for the case of two sequences, as we will then reuse this same algorithm in our MPC approach.

First, sort input pairs in $M$ by the coordinate $y$ into the sequence $M[1]$, $M[2]$, \ldots, $M[N]$, so that $M[i].y\leq M[j].y$ holds for all $i<j$. This will ensure that we consider the overlapping ranges in sequence $T$ in the correct order.  Then, we fill a table $C[1\ltdots N]$ analogous to that of Theorem~\ref{thm:colinearchainingDAGtrivial} so that $C[j]$ gives the maximum ordered coverage of $R[1 \ltdots M[j].d]$ using the pair $M[j]$ and any subset of pairs from $\{M[1],M[2],\dots, M[j-1]\}$. Hence, $\max_j C[j]$ gives the total 
maximum ordered coverage of $R$. 

Consider Equations~(\ref{eq:case-a})~and~(\ref{eq:case-b}). 
Now we can use an \emph{invariant technique} to convert these recurrence relations so that we can exploit the range maximum queries of Lemma~\ref{lemma:searchtree}:
\begin{eqnarray*}
C^\mathtt{a}[j]&=&(M[j].d-M[j].c+1) +\max_{j' \wt M[j'].d<M[j].c} C[j']\\
&=& (M[j].d-M[j].c+1)+\mathcal{T}.\mathsf{RMaxQ}(0,M[j].c-1),\\ 
C^\mathtt{b}[j]&=&M[j].d +\max_{j' \wt M[j].c\leq M[j'].d\leq M[j].d} \{C[j']-M[j'].d\} \\
&=&M[j].d+\mathcal{I}.\mathsf{RMaxQ}(M[j].c,M[j].d), \\
C[j]&=&\max(C^\mathtt{a}[j],C^\mathtt{b}[j]).
\end{eqnarray*}
For these to work correctly, we need to have properly updated the trees $\mathcal{T}$ and $\mathcal{I}$ for all $j' \in [1\ltdots j-1]$. That is, we need to call $\mathcal{T}. \mathsf{update}(M[j'].d,C[j'])$ and $\mathcal{I}.\mathsf{update}(M[j'].d,C[j']-M[j'].d)$ after computing each $C[j']$. The running time is $O(N \log N)$. 
%The pseudocode is given in Algorithm~\ref{algo:colinearchaining}.

Figure~\ref{fig:colinearchainingoutput} illustrates the optimal chain on our schematic example. This chain can be extracted by modifying the algorithm to store traceback pointers. 

\begin{theorem}[\cite{SK03,Abo07}]
\label{thm:colinearchaining}
Problem~\ref{prob:colinearchaining} on $N$ input pairs can be solved in the optimal $O(N \log N)$ time.
\end{theorem}

\subsection{Co-linear chaining on DAGs using a minimum path cover}

Let us now modify the above algorithm to work with DAGs, using the main technique of this paper.

\begin{theorem}
\label{thm:colinearchainingDAGadvanced}
Problem~\ref{prob:colinearchainingDAGlimited} on a labeled DAG $G=(V,E,\ell,\Sigma)$ of width $k$ and a set of $N$ input pairs can be solved in time $O(k|E| \log |V|+ kN \log N)$ time. 
\end{theorem}

\begin{proof}
Assume we have a path cover of size $K$ and $\forward[u]$ computed for all $u\in V$. For each path $i\in [1..K]$, we create two binary search trees $\mathcal{T}_i$ and $\mathcal{I}_i$. As a reminder, these trees correspond to coverages for pairs that do not, and do overlap, respectively, on the sequence. Moreover, recall that in Problem~\ref{prob:colinearchainingDAGlimited} we do not consider solutions where consecutive paths in the graph overlap.

As keys, we use $M[j].d$, for every pair $M[j]$, and additionally the key 0. The value of every key is initialized to $-\infty$. 

After these preprocessing steps, we process the nodes in topological order, as detailed in Algorithm~\ref{algo:colinearchainingDAGkpathcover}. If node $v$ corresponds to the endpoint of some $M[j].P$, we update the trees $\mathcal{T}_i$ and $\mathcal{I}_i$ for all covering paths $i$ containing node $v$. Then we follow all forward propagation links $(w,i) \in \forward[v]$ and update $C[j]$ for each path $M[j].P$ starting at $w$, taking into account all pairs whose path endpoints are in covering path $i$. Before the main loop visits $w$, we have processed all forward propagation links to $w$, and the computation of $C[j]$ has taken all previous pairs into account, as in the naive algorithm, but now indirectly through the $K$ search trees. Exceptions are the pairs overlapping in the graph, which we omit in this problem statement. The forward propagation ensures that the search tree query results are indeed taking only reachable pairs into account. While $C[j]$ is already computed when visiting $w$, the startpoint of $M[j].P$, the added coverage with the pair is updated to the search trees only when visiting the endpoint. 

%VM: Check the running time

There are $N K$ forward propagation links, and both search trees are queried in $O(\log N)$ time. All the search trees containing a path endpoint of a pair are updated. Each endpoint can be contained in at most $K$ paths, so this also gives the same bound $2N K$ on the number of updates. With Theorem~\ref{thm:MPC} plugged in, we have $K = k$ and the total running time becomes $O(k|E| \log |V|+k N \log N)$.
\qed
\end{proof}
%Each path contains $N_i$ of the tuple endpoints, with $\sum\limits_{i=1}^{k \log |V|} N_i = N$, which gives us a total query time of $\sum\limits_{i=1}^{k \log |V|} \log N_i = O(k \log |V| \log \frac{N}{k \log |V|})$ over all the paths for finding the maximum ordered coverage for tuple $v_j$. For $N$ tuples, the total query time is $ O(N k \log |V| \log \frac{N}{k \log |V|})$.\qed
%AK: Should this be unified such that v \in forward(v)? Can it? Queries needed to be done at start of the tuple and updates at end, but now that we don't allow overlaps in the graph, it should not matter.
%AK: O(|V|) time for going over all the nodes, O(N \log \frac{N}{k \log |V|) for the same path queries updated when meet node v, same for updating the trees for same path, then O(N k \log |V| \log \frac{N}{k \log |V|}) that we issue to tuples corresponding to forward links.

\begin{algorithm}[t]
%\dontprintsemicolon
\KwIn{DAG $G=(V,E)$, a path cover $P_1,P_2,\ldots, P_K$ of $G$, and $N$ pairs $M[1],M[2],\ldots,M[N]$ of the form $(P,[c \ltdots d])$.}
\KwOut{The index $j$ giving $\max_j C[j]$.}
Use Lemma~\ref{lemma:last2reach} to find all forward propagation links\;
\For{$i \gets 1$ to $K$}{
  Initialize search trees $\mathcal{T}_i$ and $\mathcal{I}_i$ with keys $M[j].d$, $1 \leq j \leq N$, and with key $0$, all keys associated with values $-\infty$\;
  $\mathcal{T}_i. \mathsf{update}(0,0)$\; 
  $\mathcal{I}_i. \mathsf{update}(0,0)$\; 
}
\tcc{Save to $\mathtt{start}[i]$ (respectively, $\mathtt{end}[i]$) the indexes of all pairs whose path starts (respectively, ends) at $i$.}
\For{$j \gets 1$ to $N$}{
	$\mathtt{start}[M[j].P.first]. \mathsf{push}(j)$\;
	$\mathtt{end}[M[j].P.last]. \mathsf{push}(j)$\;
}
\For{$v \in V$ in topological order}{
    \For{$j \in \mathtt{end}[v]$}{
       \tcc{Update the search trees for every path that covers $v$, stored in $\mathtt{paths}[v]$.}
       \For{$i \in \mathtt{paths}[v]$}{
          $\mathcal{T}_i. \mathsf{update}(M[j].d,C[j])$\;
          $\mathcal{I}_i. \mathsf{update}(M[j].d,C[j]-M[j].d)$\;
       }
    }
   \For{$(w,i) \in \forward[v]$}{
      \For{$j \in \mathtt{start}[w]$}{
           $C^\mathtt{a}[j] \gets (M[j].d-M[j].c+1)+\mathcal{T}_i.\mathsf{RMaxQ}(0,M[j].c-1)$\;
           $C^\mathtt{b}[j] \gets M[j].d+\mathcal{I}_i.\mathsf{RMaxQ}(M[j].c,M[j].d)$\;
           $C[j] \gets \max(C[j],C^\mathtt{a}[j],C^\mathtt{b}[j])$\;
       }
    }
%AK: C^a is RMaxQ from 0 till caller tuple's end (v_j) in the caller tuple's tree plus forward link tuple's length (v_i), but what is C^b? This should be the overlap between caller tuple v_j and forward tuple v_i plus RMax over some range.
% The saved keys were v_j.d, that is, the end coordinate in the sequence.
}

%\For{$j \gets 1$ to $N$}{
%  $C^\mathtt{a}[j]\gets (v_{j}.d-v_{j}.c+1)+\mathcal{T}.\mathsf{RMaxQ}(0,v_{j}.c-1)$\;
%  $C^\mathtt{b}[j]\gets v_{j}.d+\mathcal{I}.\mathsf{RMaxQ}(v_{j}.c,v_{j}.d)$\;
%  $C[j]\gets \max(C^\mathtt{a}[j],C^\mathtt{b}[j])$\;
%  $\mathcal{T}.\mathsf{update}(v_{j}.d,C[j])$\;
%  $\mathcal{I}.\mathsf{update}(v_{j}.d,C[j]-v_{j}.d)$\;
%}
\Return{$\argmax_j C[j]$}\;
\caption{\label{algo:colinearchainingDAGkpathcover}
Co-linear chaining between a sequence and a DAG using a path cover.}
\end{algorithm}

Appendix~\ref{sect:pathoverlaps} shows how to handle the case of path overlaps, giving the following result:

\begin{theorem}
Let $G=(V,E,\ell,\Sigma)$ be a labeled DAG and let $M$ be a set of $N$ pairs of the form $(P,[c \ltdots d])$. The algorithms from Theorems~\ref{thm:colinearchainingDAGtrivial} and \ref{thm:colinearchainingDAGadvanced} can be modified to solve Problem~\ref{prob:colinearchainingDAG} with additional time $O(L\log^2 |V|)$ or $O(L+\mathtt{\#overlaps})$, where $L$ is at most the input length and $\mathtt{\#overlaps}$ is the number of overlaps between the input paths.
\label{thm:CLC-overlaps}
\end{theorem}

The bound $O(L+\mathtt{\#overlaps})$ comes as a direct consequence of using a generalized suffix tree to compute the overlaps in advance \cite[proof of Theorem 2]{Rizzi:2014aa}. With the overlaps given, one can process each in constant time to see if they give the maximum for $C[j]$. The other bound $O(L\log^2 |V|)$ comes from backward searching a subpath in a concatenation of all subpaths using FM-index. At each step a range can be identified that contains all subpaths with suffix-prefix overlap with the current one. This range limits the keys suitable in the binary search trees storing the already computed coverage values, and thus a two-dimensional range search is needed (see Appendix~\ref{sect:pathoverlaps}).

%\alex{One reviewer commented: In Algorithm 1, there are few operations that are not explicitly defined before, e.g. start, end, and paths. (I presume that start and end stand for the lists for the pairs start/end at that particular node, and paths is the path cover?)}

\section{Discussion and experiments\label{sect:discussion}}

For applying our solutions to Problem~\ref{prob:colinearchainingDAG} in practice, one first needs to find the alignment anchors. As explained in the problem formulation, alignment anchors are such pairs $(P,[c \ltdots d])$ where $P$ is a path in $G$ and $\ell(P)$ matches $R[c \ltdots d]$. 
With sequence inputs, such pairs are usually taken to be \emph{maximal exact matches} (MEMs) and can be retrieved in small space in linear time \cite{BCMK13,Bel14}. It is largely an open problem how to retrieve MEMs between a sequence and a DAG efficiently: The case of length-limited MEMs is studied in \cite{GCSA2}, based on an extension of \cite{GCSA} with features such as suffix tree functionality. On the practical side, anchor finding has already been incorporated into tools for conducting alignment of a sequence to a DAG \cite{limasset2016read,novak2016graph}. 

\No{
However, practical approaches exist; for example, there is an index~\cite{GCSA} based on the Burrows-Wheeler transform that takes small space on the average case, under some distribution of DAGs. This index supports fast exact pattern matching, so one can retrieve all maximal exact matches by doing the backward search for all prefixes of $R$ and pruning the occurrence lists appropriately. 
More principled ways are possible for reporting length-limited MEMs~\cite{GCSA2}, by improving~\cite{GCSA} with features such as suffix tree functionality.}

For the purpose of demonstrating the efficiency of our MPC-approach applied to co-linear chaining, we implemented a MEM-finding routine based on simple dynamic programming. We leave it for future work to incorporate a practical procedure (e.g. like those in \cite{limasset2016read,novak2016graph}). 
We tested the time improvement of our MPC-approach (Theorem~\ref{thm:colinearchainingDAGadvanced}) over the trivial algorithm (Theorem~\ref{thm:colinearchainingDAGtrivial}) \alex{on the sequence graphs of annotated human genes. Out of all the 62219 genes in the HG38 annotation for all human chromosomes, we singled out 8628 genes such that their sequence graph had at least 5000 nodes. Out of these, we picked 500 genes at random.}

\alex{The size of the graphs for these 500 genes varied between $|V|=5023$ and $|V|=30959$ vertices. Their width, i.e., the number of paths in the MPC, varied between  $k=1$ and $k=15$. (The number of graphs for each value of $k$ is listed in the column \#graphs of the top table of Figure~\ref{table:running-times}.) The number of anchors, $N$, for patterns of length 1000 varied between $10^1$ and $10^5$. As shown in Figure~\ref{table:running-times}, with small values of $N$, our MPC-based co-linear chaining algorithm was twice as fast as the trivial algorithm. When values of $N$ were increased from $10^1$ to $10^5$, the difference increased to two orders of magnitude.}

% AK: Should we say something about the fact that third row has much lower vertex count? There is no case where |V| would be in tens of thousands and N in thousands.

%\begin{table}
%\centering
%\begin{tabular}{| r | r | r | r | r |}
%\hline
%MPC method & Trivial & $|V|$ & $k$ & $N$\\
%\hline
%9 & 94 & 12,564 & 5 & 28\\
%13 & 1,686 & 14,914 & 3 & 453\\
%48 & 7,290 & 5,254 & 5 & 4,610\\
%272 & 198,658 & 26,544 & 4 & 29,855\\
%\hline
%\end{tabular}
%\caption{The running time comparison (in milliseconds) of the two approaches. Both approaches are given the same anchors and the preprocessing time for finding the anchors is not included in running times. \label{table:running-times}}
%\end{table}

\begin{figure}[t]
\footnotesize
\centering
\begin{tabular}{|r|r|r|r|r|}
\hline
~~$k$ & \#graphs & mean $|V|$ & MPC method & Naive method\\
\hline
1 & 75 & 7275 & 18 {\footnotesize $\pm$ 27}ms & 5638 {\footnotesize $\pm$ 12378}ms\\
2 & 117 & 8109 & 23 {\footnotesize $\pm$ 36}ms & 6355 {\footnotesize $\pm$ 17641}ms\\
3 & 93 & 8306 & 27 {\footnotesize $\pm$ 41}ms & 6499 {\footnotesize $\pm$ 17940}ms\\
4 & 99 & 8933 & 32 {\footnotesize $\pm$ 49}ms & 6864 {\footnotesize $\pm$ 17868}ms\\
5 & 48 & 9779 & 40 {\footnotesize $\pm$ 59}ms & 8053 {\footnotesize $\pm$ 18742}ms\\
6 & 32 & 10265 & 45 {\footnotesize $\pm$ 65}ms & 7934 {\footnotesize $\pm$ 16659}ms\\
7 & 16 & 9928 & 41 {\footnotesize $\pm$ 59}ms & 6973 {\footnotesize $\pm$ 15345}ms\\
8 & 10 & 11052 & 57 {\footnotesize $\pm$ 83}ms & 8731 {\footnotesize $\pm$ 17497}ms\\
9 & 4 & 9538 & 52 {\footnotesize $\pm$ 77}ms & 6252 {\footnotesize $\pm$ 13906}ms\\
10 & 3 & 10833 & 61 {\footnotesize $\pm$ 102}ms & 7055 {\footnotesize $\pm$ 16221}ms\\
11 & 2 & 11186 & 50 {\footnotesize $\pm$ 70}ms & 5932 {\footnotesize $\pm$ 10548}ms\\
15 & 1 & 16848 & 154 {\footnotesize $\pm$ 194}ms & 25253 {\footnotesize $\pm$ 43873}ms\\\hline
\end{tabular}\\
\vspace{.5cm}
\begin{tabular}{|r|r|r|r|}
\hline
~~$N$ & mean $|V|$ & MPC method & Naive method\\
\hline
$(10^{0}..10^{1}]$ & 8681 & 8 {\footnotesize $\pm$ 5}ms & 15 {\footnotesize $\pm$ 8}ms\\
$(10^{1}..10^{2}]$ & 8808 & 8 {\footnotesize $\pm$ 5}ms & 79 {\footnotesize $\pm$ 68}ms\\
$(10^{2}..10^{3}]$ & 9732 & 10 {\footnotesize $\pm$ 7}ms & 524 {\footnotesize $\pm$ 392}ms\\
$(10^{3}..10^{4}]$ & 6824 & 70 {\footnotesize $\pm$ 22}ms & 15153 {\footnotesize $\pm$ 5875}ms\\
$(10^{4}..10^{5}]$ & 12235 & 153 {\footnotesize $\pm$ 66}ms & 49482 {\footnotesize $\pm$ 31900}ms\\
\hline
\end{tabular}

\caption{The average running times, and their standard deviation, (in milliseconds) of the two approaches for co-linear chaining between a sequence and a DAG (Problem~\ref{prob:colinearchainingDAG}), for all inputs of a certain width $k$ (top), and with $N$ belonging to a certain interval (below). Both approaches are given the same anchors; the time for finding them is not included. \label{table:running-times}}
\end{figure}

The improved efficiency when compared to the naive approach gives reason to believe a practical sequence-to-DAG aligner can be engineered along the algorithmic foundations given here. Future work includes the incorporation of a practical anchor-finding method, and testing whether the complete scheme improves transcript prediction through improved finding of exon chains \cite{KTM17}.

\alex{On the theoretical side, it remains open whether the MPC algorithm could benefit from a better initial approximation and/or one that is faster to compute. More generally, it remains open whether the overall bound $O(k|E|\log|V|)$ for the MPC problem can be improved.}

%The trivial approach of graph alignment is not suitable for this purpose, as finding the anchors using semi-local alignment would take as much time as finding the whole alignment (with suitable gap penalty score).

%On practical front there exists heuristics for searching in two graphs using Burrows-Wheeler -transform and de Bruijn graphs~\cite{GCSA,GCSA2}. When such data structures are augmented with suffix tree functionality~\cite{GCSA2}, they allow finding e.g. maximal exact matches in way similar to finding maximal exact matches in sequences using generalized suffix trees. However, currently the implementation only supports variation graphs and the size of the queries is limited to the order of the path graph.

\paragraph{Acknowledgements.}

We thank Djamal Belazzougui for pointers on backward step on large alphabet and Gonzalo Navarro for pointing out the connection to pattern matching on hypertexts. This work was funded in part by the Academy of Finland (grant 274977 to AIT and grants 284598 and 309048 to AK and to VM), and by Futurice Oy (to TP).

\bibliographystyle{plain}
\bibliography{biblio}

\newpage
\appendix

\section{Two simple applications}

\subsection{Reachability queries \label{appendix:reachability}}

An immediate application of Theorem~\ref{thm:MPC} and of the values $\lasttoreach[v,i]$ is for solving reachability queries. If we have all these $O(k|V|)$ values, then we can answer in constant time whether a node $y$ is reachable from a node $x$, as in \cite{Jagadish:1990:CTM:99935.99944}: we check $\index[x,i] \leq \index[\lasttoreach[y,i],i]$, where $\index$ was defined in the proof of Lemma~\ref{lemma:last2reach}, $i$ is a path containing $x$, and we take by convention $\index[-1,i] = -1$. Recall also that reachability queries in an arbitrary graph can be reduced to solving reachability queries in its DAG of strongly connected components, because nodes in the same component are pairwise reachable. See Table~\ref{tab:reachbility} for existing tradeoffs for solving reachability queries.\footnote{Note that \cite{chain-cover} incorrectly attributes to \cite{Jagadish:1990:CTM:99935.99944} query time $O(\log k)$, and as a consequence \cite{7750623,Jin:2011:PER:1929934.1929941} incorrectly mention query time $O(\log k)$ for \cite{chain-cover}.}

\begin{corollary}
\label{cor:reachbility}
Let $G = (V,E)$ be an arbitrary directed graph and let the width of its DAG of strongly connected components be $k$. In time $O(k|E|\log |V|)$ we can construct from $G$ an index of size $O(k|V|)$, so that for any $x,y \in V$ we can answer in~$O(1)$~time whether $y$ is reachable from $x$.
\end{corollary}

\begin{table}[t!]
\begin{center}
\renewcommand{\arraystretch}{1.5}
\begin{tabular}{|c|c|c|c|}\hline
Construction time & Index size & Query time & Reference \\\hline 
$O(k|E|\log |V|)$ & $O(k|V|)$ & $O(1)$ & this paper\\\hline 
$O(|V|^2 + k\sqrt{k}|V|)$ & $O(k|V|)$ & $O(1)$ & \cite{chain-cover}\\\hline 
$O(k|E|)$ or $O(|V||E|)$ & $O(k|V|)$ & $O(\log^2 k)$ & \cite{Jin:2011:PER:1929934.1929941}\\\hline
$O(k(|V| + |E|))$ & $O(k|V|)$ & $O(k)$ or $O(|V| + |E|)$ & \cite{Yildirim:2010:GSR:1920841.1920879}\\
\hline 
\end{tabular} 
\end{center}
\caption{Previous comparable space/time tradeoffs for solving reachability queries. Compiled from~\cite[Table 1]{7750623}.\label{tab:reachbility}}
\end{table}

\subsection{The LIS problem \label{appendix:LIS}}
\label{sec:lis}

The LIS problem asks us to delete the minimum number of values from an input sequence $s_1\cdots s_n$ such that remaining values form a strictly increasing series of values. 
Here the input sequence is assumed to come from an ordered alphabet $\Sigma$. For example, on input sequence $1,4,2,3,7,5,6$, from the alphabet $\Sigma=\{1,2,3,4,5,6,7\}$, the unique optimal solution is $1,2,3,5,6$. Such a longest increasing subsequence can be found in the optimal $O(n\log n)$ time~\cite{computinglis}. 

This optimal algorithm works as follows. We first map $\Sigma$ to a subset of $\{1,2,\ldots,n\}$ with an order-preserving mapping, in $O(n \log n)$ time (by e.g., sorting the sequence elements, and relabeling by the ranks in the order of distinct values). We then store, at every index $i$ of the input sequence, the value $\mathtt{LLIS}[i]$ defined as the length of the longest strictly increasing subsequence ending at $i$ and using the $i$-th symbol. 
The values $\mathtt{LLIS}[i]$ can be computed by dynamic programming, by storing all previous key-value pairs $(s_j,\mathtt{LLIS}[j])$ in a search tree $\mathcal{T}$ as in Lemma~\ref{lemma:searchtree}, and querying $\mathcal{T}.\mathsf{RMaxQ}(0,s_i-1)$. 

Consider the following extension of the LIS problem to a labeled DAG $G=(V,E,\ell,\Sigma)$ of width $k$. For a path $P = (v_1,\dots,v_t)$ in $G$, let the \emph{label} of $P$, denoted $\ell(P)$, be the concatenation of the labels of the nodes of $P$, namely $\ell(v_1)\cdots\ell(v_t)$. Among all paths $P$ in $G$, and among all subsequences of $\ell(P)$, we need to find a longest strictly increasing subsequence. 

We now explain how to extend the previous dynamic programming algorithm for this problem. We analogously map $\Sigma$ to a subset of $\{1,2,\ldots,|V|\}$ with an order-preserving mapping in $O(|V| \log |V|)$ time, as above. Recall that we assume $V = \{1,\dots,|V|\}$, where $1,\dots,|V|$ is a topological order. Assume also that we have $K$ paths to cover $V$ and $\forward[u]$ is computed for all $u\in V$.

For each node $v$, we aim to analogously compute $\mathtt{LLIS}[v]$ as the length of a longest strictly increasing subsequence of the labels of all paths ending at $v$, with the property that $\ell(v)$ is the last element of this subsequence. 

\begin{sloppypar}
For each $i \in [1..K]$, we let $\mathcal{T}_i$ be a search tree as in Lemma~\ref{lemma:searchtree}, initialized with key-value pairs $(0,0),(1,-\infty),(2,-\infty),\ldots, (|V|,-\infty)$. The algorithm proceeds in the fixed topological ordering. Assume now that we are at some position $u$, and have already updated all search trees associated with the  covering paths going through $u$. For every $(v,i) \in \forward[u]$, we update $\mathtt{LLIS}[v]=\max(\mathtt{LLIS}[v],\mathcal{T}_i.\mathsf{RMaxQ}(0,\ell(v)-1)+1)$. Once the algorithm reaches $v$ in the topological ordering, value $\mathtt{LLIS}[v]$ has been updated from all $u'$ such that $(v,i) \in \forward[u']$. It remains to show how to update each $\mathcal{T}_i$ when reaching $v$, for all covering paths $i$ on which $v$ occurs. This is done as $\mathcal{T}_{i}.\mathsf{update}(\ell(v),\mathtt{LLIS}[v])$. Initialization is handled by the $(0,0)$ key-value pair so that any position can start a new increasing subsequence.  Figure~\ref{fig:example} shows an example.
\end{sloppypar}

\begin{figure}
\centering
\includegraphics[width=7cm]{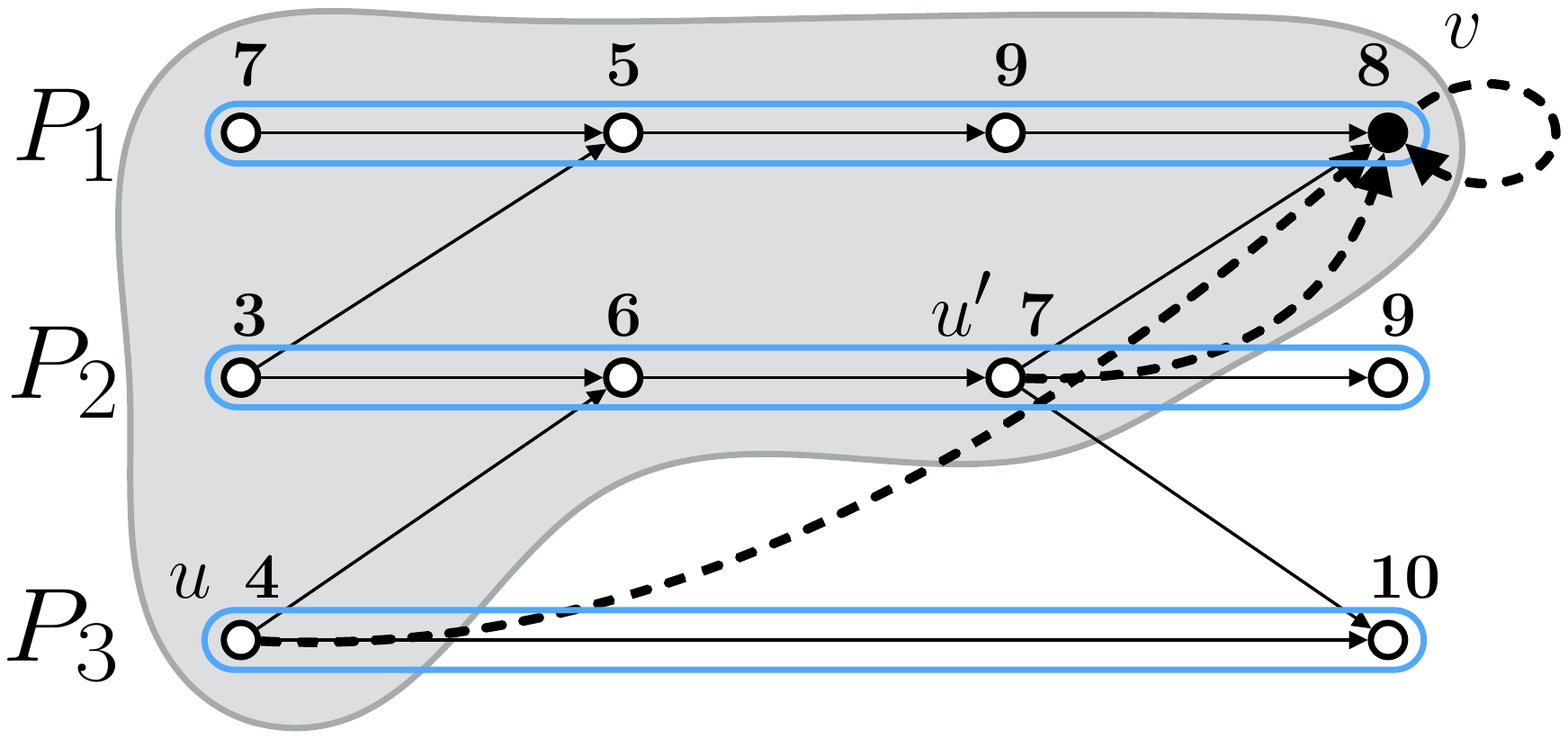}
\caption{Detailed steps for solving a LIS DAG instance, focusing on three nodes: $u$, $u'$, and $v$. In any topological order, node $u$ will be processed before node $u'$, which in turn will be processed before $v$. After computing $\mathtt{LLIS} [u] = 1$ we set $\mathtt{LLIS} [v] = \max (1, \mathcal{T}_3.\mathsf{RMaxQ} (0, 7) + 1) = 2$, since at that point $\mathcal{T}_3$ contains $(0, 0), (1, -\infty), \ldots, (3, -\infty), (4, 1), (5, -\infty), \ldots, (10, -\infty)$.  After computing $\mathtt{LLIS} [u'] = 3$, we set $\mathtt{LLIS} [v] = \max (2, \mathcal{T}_2.\mathsf{RMaxQ} (0, 7) + 1) = 4$, since at that point $\mathcal{T}_2$ contains $(0, 0), (1, -\infty), (2, -\infty), (3, 1), (4, -\infty), (5, -\infty), (6, 2), (7, 3), (8, -\infty), \ldots, (10, -\infty)$.  When we finally reach $v$ we leave $\mathtt{LLIS} [v]= \max (4, \mathcal{T}_1.\mathsf{RMaxQ} (0, 7) + 1) = 4$, since at that point $\mathcal{T}_1$ contains $(0,0), (1, -\infty), \ldots, (4, -\infty), (5, 2), (6, -\infty), (7, 1), (8, -\infty), (9, 3), (10, -\infty)$.
.\label{fig:example}}
\end{figure}

The final answer to the problem is $\max_{v\in V} \mathtt{LLIS}[v]$, with the actual LIS to be found with a standard traceback. The algorithm runs in $O(K |V| \log |V|)$ time. With Theorem~\ref{thm:MPC} plugged in, we have $K = k$ and the total running time becomes $O(k|E| \log |V|+k|V|\log |V|) =  O(k|E| \log |V|)$, under our assumption $|E| \geq |V| - 1$. The following theorem summarizes this result.

\begin{theorem}
Let $G = (V,E,\ell,\Sigma)$ be a labeled DAG of width $k$, where $\Sigma$ is an ordered alphabet. We can find a longest increasing subsequence in $G$ in time $O(k|E| \log |V|)$.
\end{theorem}

When the DAG is just a labeled path with $|E| = |V| - 1$ (modeling the standard LIS problem), then the algorithm from Lemma~\ref{lemma:approx-path-cover} returns one path ($K = 1$). The complexity is then $O(|V|\log|V|)$, matching the best possible bound for the standard LIS problem~\cite{computinglis}.

\section{Co-linear chaining with path overlaps\label{sect:pathoverlaps}}

We now consider how to extend the algorithms we developed for Problem~\ref{prob:colinearchainingDAGlimited} to work for the more general case of Problem~\ref{prob:colinearchainingDAG}, where overlaps between paths are allowed in a solution. The detection and merging of such path overlaps has been studied in \cite{Rizzi:2014aa}, and we tailor a similar approach for our purposes.

We use an \emph{FM-index} \cite{FM05} tailored for large alphabets \cite{HSS09}, and a two-dimensional range search tree \cite{BCKO08} modified to support range maximum queries. The former is used for obtaining all ranges $[i'..i]$ in the coverage array $C$ such that all input pairs $M[i'],\ldots, M[i]$ have a path $M[i''].P$, $i'\leq i''\leq i$, overlapping with the path $M[j].P$ of $j$-th input pair $M[j]$. Here, the endpoint of the $j$-th input pair is at node $v$ visited in topological order. This implies that the paths of the input pairs $[i'..i]$ have already been visited, and thus, by induction, that $C[i'..i]$ values have been correctly computed (subject to the modification we are about to study). The sequence ranges $[M[i''].c..M[i''].d]$  for all $i''\in [i'..i]$ may be arbitrarily located with respect to interval $[M[j].c..M[j].d]$, so we need to maintain an analogous mechanism with search trees of type $\mathcal{T}$ and $\mathcal{I}$ as in our co-linear chaining algorithm based on a path cover. This time we cannot, in advance, separate the input pairs to $K$ paths with different search trees, but we have a dynamic setting, with interval $[i'..i]$ deciding which values should be taken into account. This is where a two-dimensional range search tree is used to support these queries in $O(\log^2 N)$ time: Figure~\ref{fig:2dsearchtree} illustrates this.  

\begin{figure}[h!]
\centering
\includegraphics[width=9cm]{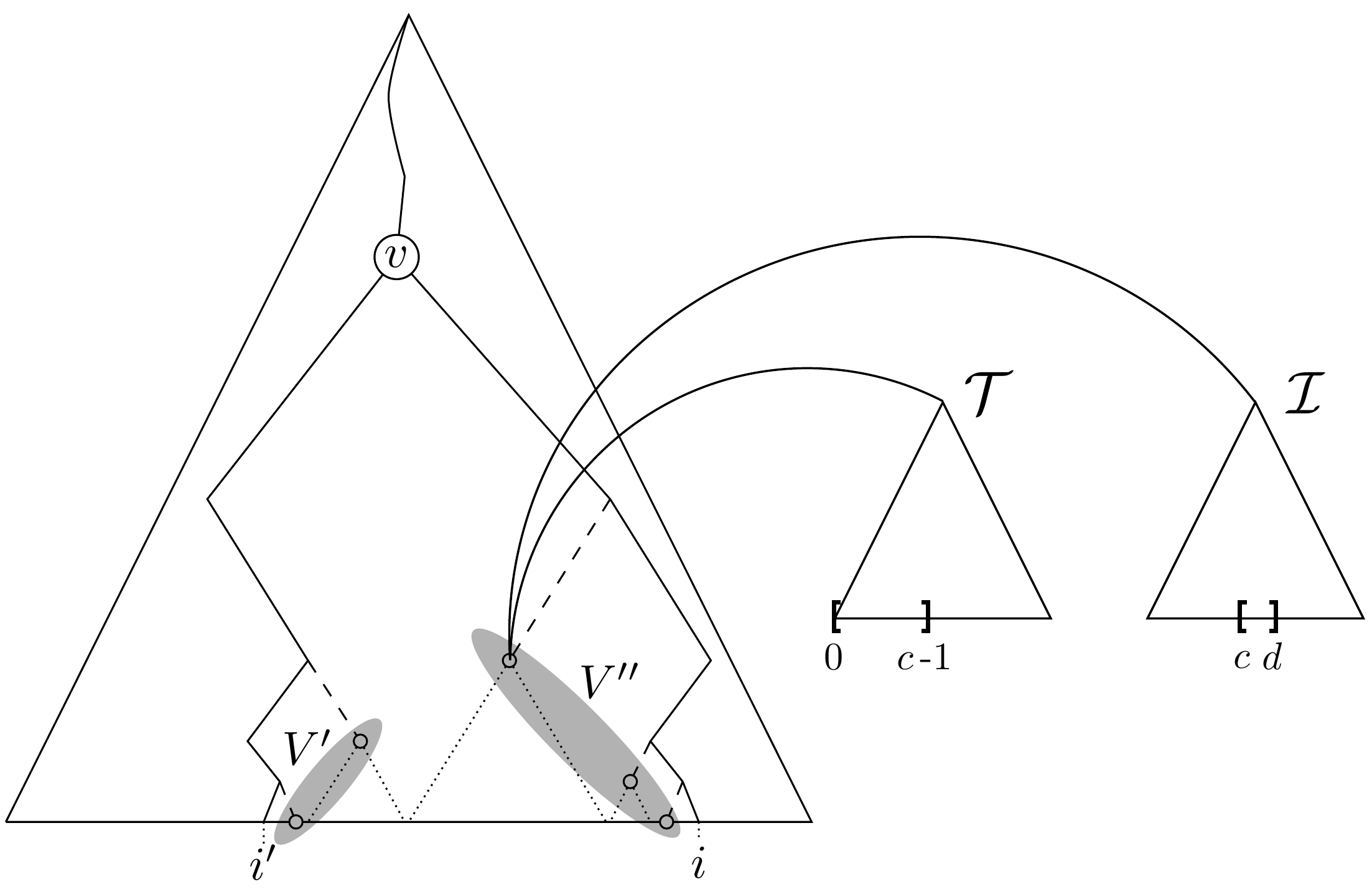}
\caption{Answering two-dimensional range maximum queries. A search interval $[i'..i]$ is split to $O(\log (i-i'+1))$ intervals / subtrees. In each subtree root we store a type $\mathcal{T}$ and a type $\mathcal{I}$ search tree of Lemma~\ref{lemma:searchtree} containing values $C[i'']$ or $C[i'']-M[i''].d$, respectively, for indexes $i''$ in the corresponding subinterval.  
\label{fig:2dsearchtree}}
\end{figure}

In what follows we show that $O(L)$ queries are sufficient to take all overlaps into account throughout the algorithm execution (holding for both the trivial algorithm and for the one based on a path cover), where $L=\sum_i |M[i].P|$---the sum of the path lengths---is at most the total input length. The construction will actually induce an order for the input pairs such that $O(L)$ queries are sufficient: Since the other parts of the algorithms do not use the order of input pairs directly, we can safely reorganize the input accordingly.  

\begin{sloppypar}
With this introduction, we are ready to consider how all the intervals $[i'..i]$ related to $j$-th pair are obtained. 
We build in $O(L \log\log  |V|)$ time the FM-index version proposed in \cite{HSS09} of sequences $T=(\prod_i \#(M[i].P)^{-1})\#$, where $\#$ is a symbol not in alphabet $\{1,2,\ldots,|V|\}$ and considered smaller than other symbols, e.g. $\#=0$, and  $X^{-1}$ denotes the reverse $X[|X|]X[|X-1]]\cdots X[1]$ of $X$.
\end{sloppypar}

For our purposes it is sufficient to know that the FM-index of $T$, when given an interval $I(X)$ corresponding to lexicographically-ordered suffixes that start with $X$, can determine the interval $I(cX)$ in $O(\log\log|V|)$ time \cite{HSS09}.
This operation is called \emph{backward step}.

\begin{sloppypar}
We use the index to search $M[j].P$ in the forward direction by searching its reverse with backward steps. Consider we have found interval $I((M[j].P[1..k])^{-1})$, for some $k$, such that backward step $I(\#(M[j].P[1..k])^{-1})$ results in a non-empty interval $[i'..i]$. This interval $[i'..i]$ corresponds to all suffixes of $T$ that have $\#(M[j].P[1..k])^{-1}$ as a prefix. That is, $[i'..i]$ corresponds to input pairs whose path suffix have a length $k$ overlap with the path prefix of $j$-th input pair. For this interval to match with coverage array $C$, we just need to rearrange the input pairs according to their order in the first $N$ rows of the array storing the lexicographic order of suffixes of $T$.
\end{sloppypar}

Since each backward step on the index may induce a range search on exactly one interval $[i'..i]$, the running time is dominated by the range queries.\footnote{A simple wavelet tree based FM-index would provide the same bound, but in case the range search part is later improved, we used the best bound for the subroutine.} On the other hand, this also gives the bound $L$ on the number of range queries, as claimed earlier. 

Alternatively, one can omit the expensive range queries and process each overlapping pair separately, to compute in constant time its contribution to $C[j]$. This gives another bound $O(L\log\log |V|+\mathtt{\#overlaps})$, where $\mathtt{\#overlaps}$ is the number of overlaps between the input paths. This can be improved to $O(L+\mathtt{\#overlaps})$ by using a generalized suffix tree to compute the overlaps in advance \cite[proof of Theorem 2]{Rizzi:2014aa}.

The result is summarized in Theorem~\ref{thm:CLC-overlaps} at page~\pageref{thm:CLC-overlaps}.

\end{document}